\documentclass[conference]{IEEEtran}
\IEEEoverridecommandlockouts
\usepackage{cite}
\usepackage{amsmath,amssymb,amsfonts}
\usepackage[ruled,vlined,linesnumbered]{algorithm2e}
\usepackage{graphicx}
\usepackage{url}
\usepackage{textcomp}
\usepackage{amsthm}

\newtheorem{proposition}{Proposition}
\begin{document}

\title{Outage-Aware Robust Sensing for UAV ISAC Systems
\thanks{Accepted by IEEE GLOBECOM 2026 Workshop}}

\author{\IEEEauthorblockN{
        Qiming~Li,
        Lei~Zhang,
        Haoran~Xu
        and~Lina~Mohjazi}
        \IEEEauthorblockA{
        School of Engineering, University of Glasgow, G12 8QQ, U.K.\\
        Email: 
            q.li.5@research.gla.ac.uk,
            Lei.Zhang@glasgow.ac.uk,
            h.xu.3@research.gla.ac.uk,
            Lina.Mohjazi@glasgow.ac.uk
        }

}

\maketitle

\begin{abstract}
Integrated sensing and communications (ISAC) infrastructure can provide an external perception layer for aerial robots, but dynamic blockage makes the set of informative sensing links at the next update uncertain. A single predicted availability pattern can miss low-information outcomes, whereas protecting all patterns wastes power and may become infeasible when all links are blocked. We propose an outage-aware conformal robust power-allocation framework for multi-BS UAV tracking. A history-aware joint predictor maps radar and tracking histories to next-slot availability probabilities, and adaptive prediction sets retain plausible correlated patterns. Tracking constraints are enforced only for retained non-outage patterns, while the all-blocked pattern is treated separately as a physical sensing outage. With per-slot beam directions fixed, the resulting scalar-power allocation is a semidefinite program. A reliability bound relates serviceable-slot tracking failure to deployed-policy miscoverage and optimization infeasibility. Across paired closed-loop simulations over five random seeds, the method attains a 0.01\% serviceable-slot constraint-violation rate and uses 47.0\% less power than protection over all non-outage patterns.
\end{abstract}


\section{Introduction}
Unmanned aerial vehicles (UAVs) are aerial robots used in logistics, inspection, and emergency response, while unauthorized UAVs motivate airspace surveillance. These settings require sensing and tracking of UAV position and motion for safe operation and situational awareness. Integrated sensing and communications (ISAC)-enabled cellular infrastructure can obtain these states by reusing spectrum and hardware \cite{liu2022isac,li2026distributed}. Cooperative multi-base-station (BS) sensing adds spatial diversity \cite{xie2022collaborative,xie2023networked}, but only informative BS--UAV links contribute to tracking. Under mobility and dynamic blockage, availability at slot $t$ may change by slot $t+1$, while power selected at slot $t$ determines next-slot observation quality. The allocator must decide before the next joint pattern is revealed, and common blockers make current-state or independent per-link forecasts insufficient.

Existing ISAC allocation methods optimize waveforms, beamformers, or powers under communication and sensing constraints \cite{liu2018mimo,bazzi2023outage,attiah2024duality}, but generally assume sensing-information states known at decision time. Blockage predictors use wireless or visual histories for proactive handoff \cite{alkhateeb2018blockage,charan2021vision}, including federated and vision-aided prediction \cite{alquraan2023federated,alquraan2022intelligent}, but do not propagate future joint-state uncertainty into tracking-constrained allocation. The interface is also closed-loop because the current sensing action changes the observations later used by the tracker and predictor.

Prediction uncertainty cannot be safely reduced to one surrogate. A maximum a posteriori (MAP) rule may omit a plausible low-information pattern, whereas expected Fisher information matrix (FIM) feasibility does not ensure per-realization feasibility. Conformal prediction supplies finite-sample marginal prediction sets \cite{romano2020classification} and has wireless applications \cite{cohen2022calibrating}. Yet enforcing every retained pattern can make the problem infeasible when the set contains the all-blocked realization: its zero FIM means power cannot restore a violated tracking threshold. Set-valued uncertainty must therefore preserve correlated joint outcomes while separating serviceable patterns from physical outage.

The resulting gap is an uncertainty interface that retains correlated next-slot patterns, passes their calibrated uncertainty to robust allocation, and excludes the power-insensitive all-blocked state from enforceable tracking constraints. 

We close this gap with a serviceability-conditioned conformal robust framework. A history-aware joint predictor and adaptive prediction set (APS) calibration retain plausible patterns. Serviceable members enter robust FIM constraints, while all-blocked is retained for warning and accounting. Fixed per-slot directions yield a scalar-power semidefinite program (SDP), and a coverage-to-tracking bound separates miscoverage, infeasibility, and outage.

The main contributions are summarized as follows.
\begin{itemize}
\item We develop a history-aware joint availability predictor and an APS-calibrated pattern set that retains correlated next-slot uncertainty.
\item We separate serviceable patterns from all-blocked outage and cast tracking-constrained scalar-power allocation as an SDP.
\item We derive a coverage-to-tracking reliability bound and evaluate the framework against point, expectation-based, full-set, all-serviceable, and oracle policies in paired closed-loop simulations.
\end{itemize}

\section{System Model}
\subsection{Network and Tracking Model}

\begin{figure*}[!t]
\centering
\includegraphics[width=\textwidth]{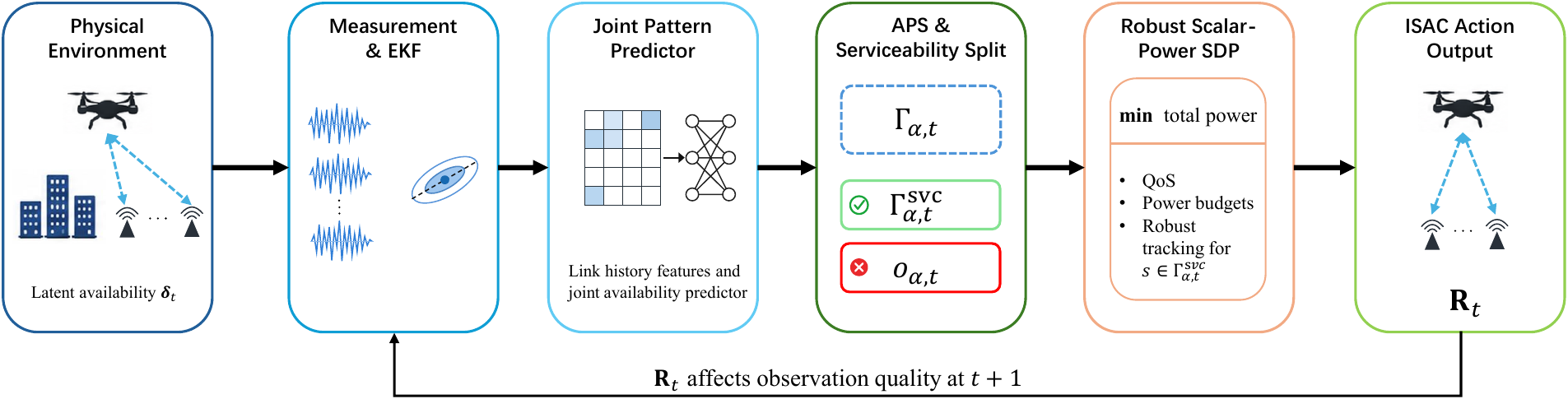}
\caption{Closed-loop data flow of the proposed serviceability-conditioned UAV ISAC framework.}
\label{fig:system_overview}
\end{figure*}

Consider an ISAC downlink in which $B$ sensing-capable base stations (BSs) serve $K$ downlink users while tracking a UAV that does not receive communication service. Let $\mathcal{B}=\{1,\ldots,B\}$ and $\mathcal{K}=\{1,\ldots,K\}$. The UAV horizontal state at slot $t$ is
\begin{equation}
\mathbf{x}_t=[q_{x,t},q_{y,t},v_{x,t},v_{y,t}]^\mathsf{T},
\end{equation}
where $(q_{x,t},q_{y,t})$ and $(v_{x,t},v_{y,t})$ are the horizontal position and velocity. A nearly constant velocity (NCV) model is used here:
\begin{equation}
\mathbf{x}_{t+1}=\mathbf{F}\mathbf{x}_t+\mathbf{w}_t,\qquad
\mathbf{w}_t\sim\mathcal{N}(\mathbf{0},\mathbf{Q}).
\end{equation}
Here, $\mathbf{F}$ is the NCV state-transition matrix and $\mathbf{Q}$ is the process-noise covariance. The tracker utilizes a standard extended Kalman filter (EKF) \cite{barshalom2001estimation}, where $\hat{\mathbf{x}}_{t|t}$ and $\mathbf{P}_{t|t}$ denote the posterior state estimate and covariance, respectively. Its prediction step gives
\begin{align}
\hat{\mathbf{x}}_{t+1|t} &= \mathbf{F}\hat{\mathbf{x}}_{t|t},\\
\mathbf{P}_{t+1|t} &= \mathbf{F}\mathbf{P}_{t|t}\mathbf{F}^\mathsf{T}+\mathbf{Q}.
\end{align}
Fig.~\ref{fig:system_overview} summarizes the causal closed-loop data flow of the proposed system.

\subsection{Sensing-Link Availability}
For BS $b$, define the binary sensing availability $\delta_{b,t}\in\{0,1\}$, where $\delta_{b,t}=1$ means that the BS--UAV sensing link can produce an effective radar update under the current propagation state, and $\delta_{b,t}=0$ means that severe blockage or an insufficient echo budget prevents a reliable update. To decouple the availability label from the optimizer’s current power decision, let $\mathrm{SCNR}^{\max}_{b,t}$ denote the
signal-to-clutter-plus-noise ratio (SCNR) evaluated at the maximum admissible probing power. We define:
\begin{equation}
\delta_{b,t}=\mathbb{I}\left\{\mathrm{SCNR}^{\max}_{b,t}\geq \gamma_{\rm det}\right\},
\label{eq:availability}
\end{equation}
where $\mathbb{I}\{\cdot\}$ is the indicator function and $\gamma_{\rm det}$ is the minimum SCNR threshold for an effective radar update. Thus, $\delta_{b,t}$ represents whether the environment permits link $b$ to support useful sensing under the admissible power budget, rather than whether the current allocation happens to illuminate it sufficiently. Missed detections and false alarms affect only the realized radar observations introduced in Sec.~\ref{sec:radar_obs}.

The joint availability vector $\boldsymbol{\delta}_t=[\delta_{1,t},\ldots,\delta_{B,t}]^\mathsf{T}$ takes one of $2^B$ binary patterns, indexed by the finite pattern set $\mathcal{S}=\{0,\ldots,2^B-1\}$. For each $s\in\mathcal{S}$, let $\boldsymbol{\delta}^{(s)}=[\delta_1^{(s)},\ldots,\delta_B^{(s)}]^\mathsf{T}$ denote the corresponding binary pattern. The realized pattern index $S_t\in\mathcal{S}$ is defined by $\boldsymbol{\delta}_t=\boldsymbol{\delta}^{(S_t)}$.

\subsection{Serviceable Patterns and Sensing Outage}
The joint pattern set is divided according to physical serviceability. The serviceable subset is
\begin{equation}
\mathcal{S}_{\rm svc}=
\left\{
s\in\mathcal{S}:
\sum_{b=1}^{B}\delta_b^{(s)}\geq 1
\right\}.
\label{eq:svc_set}
\end{equation}
The all-blocked outage pattern is
\begin{equation}
s_{\rm out}=0,\qquad
\boldsymbol{\delta}^{(s_{\rm out})}=[0,\ldots,0]^\mathsf{T} .
\label{eq:outage_pattern}
\end{equation}
Patterns in $\mathcal{S}_{\rm svc}$ contain at least one potentially informative link, whereas $s_{\rm out}$ cannot produce a valid radar update under the admissible power budget. It is therefore treated as a sensing outage rather than a beamforming failure.

\subsection{Propagation, Blockage, and Channel Models}
For each BS--UAV sensing link, let $\xi_{b,t}\in\{\mathsf{L},\mathsf{P},\mathsf{B}\}$ denote the blockage state, corresponding to line-of-sight, partial blockage, and complete blockage. This state captures the environment-induced visibility process and may exhibit both temporal and cross-link correlation. The associated attenuation factor is $\kappa(\xi_{b,t})$, with $1=\kappa(\mathsf{L})>\kappa(\mathsf{P})>\kappa(\mathsf{B})\geq0$.

Let $r_{b,t}=\|\mathbf{q}_t-\mathbf{q}_{b}\|_2$ be the three-dimensional BS--UAV distance, where $\mathbf{q}_t=[q_{x,t},q_{y,t},h_{\rm u}]^\mathsf{T}$, $h_{\rm u}$ is the UAV altitude, and $\mathbf{q}_b$ is the three-dimensional position of BS $b$. Following the radar equation \cite{skolnik2001radar}, the maximum SCNR used in \eqref{eq:availability} is 
\begin{equation}
\mathrm{SCNR}^{\max}_{b,t}=
\frac{P_{b,\rm probe}^{\max}G_{b}^{\rm tx}G_{b}^{\rm rx}\lambda_c^2\sigma_{\rm rcs}
\kappa(\xi_{b,t})}
{(4\pi)^3r_{b,t}^4\left(N_0W_{\rm r}F_b+C_{b,t}\right)},
\label{eq:scnr}
\end{equation}
where $P_{b,\rm probe}^{\max}$ is the maximum admissible probing power, $G_b^{\rm tx}$ and $G_b^{\rm rx}$ are the sensing transmit and receive gains, $\lambda_c$ is the wavelength, $\sigma_{\rm rcs}$ is the UAV radar cross section, $N_0W_{\rm r}F_b$ is receiver noise power, and $C_{b,t}$ is residual clutter.

The downlink channel from BS $m$ to user $k$ is
\begin{subequations}
    \begin{align}
        \mathbf{h}_{m,k,t}&=\sqrt{\beta_{m,k,t}}\,\widetilde{\mathbf{h}}_{m,k,t},
        \label{eq:comm_channel}\\
        \widetilde{\mathbf{h}}_{m,k,t}
        &=\sqrt{\frac{\varrho_{m,k,t}}{\varrho_{m,k,t}+1}}
        \mathbf{a}_{m,k} +
        \sqrt{\frac{1}{\varrho_{m,k,t}+1}}
        \mathbf{R}_{m,k}^{1/2}\mathbf{w}_{m,k,t},
    \end{align}
\end{subequations}
where $\beta_{m,k,t}$ follows a distance-dependent path-loss and shadowing model with blockage attenuation, $\varrho_{m,k,t}$ is the Rician factor, $\mathbf{a}_{m,k}$ is the deterministic line-of-sight array response of BS $m$, $\mathbf{R}_{m,k}$ is the spatial covariance matrix of the scattered component, and $\mathbf{w}_{m,k,t}\sim\mathcal{CN}(\mathbf{0},\mathbf{I})$. The full network channel state information (CSI) $\mathbf{h}_{k,t}$ stacks $\{\mathbf{h}_{m,k,t}\}_{m=1}^{B}$, which is also available to the resource allocator.
\addtolength{\topmargin}{0.05 in}
\addtolength{\textheight}{-0.05 in}
\subsection{Radar Observations and Closed-Loop Timing}
\label{sec:radar_obs}
When link $b$ is available, its radar measurement is modeled as
\begin{equation}
\mathbf{z}_{b,t}=\mathbf{g}_b(\mathbf{x}_t)+\mathbf{n}_{b,t},\qquad
\mathbf{n}_{b,t}\sim\mathcal{N}(\mathbf{0},\boldsymbol{\Sigma}_{b,t}).
\label{eq:measurement}
\end{equation}
The measurement function is
$\mathbf{g}_b(\mathbf{x}_t)=
\begin{bmatrix}
r_{b,t} & \phi_{b,t} & \dot r_{b,t}
\end{bmatrix}^\mathsf{T}$,
where
\begin{subequations}
    \begin{align}
        r_{b,t} &= \|\mathbf{q}_t-\mathbf{q}_b\|_2,\\
        \phi_{b,t} &= \operatorname{atan2}(q_{y,t}-q_{b,y},q_{x,t}-q_{b,x}),\\
        \dot r_{b,t} &=
        \frac{(q_{x,t}-q_{b,x})v_{x,t}+(q_{y,t}-q_{b,y})v_{y,t}}{r_{b,t}} .
    \end{align}
\end{subequations}
The EKF uses the Jacobian $\mathbf{H}_b=\partial\mathbf{g}_b/\partial\mathbf{x}$ evaluated at the predicted state. When $\delta_{b,t}=0$, BS $b$ contributes no valid measurement update. When $\delta_{b,t}=1$, the inverse measurement covariance increases with the previous-slot illumination:
\begin{equation}
\boldsymbol{\Sigma}_{b,t}^{-1}
=\zeta_{b,t}\mathbf{D}_b,\qquad
\mathbf{D}_b\triangleq
\operatorname{diag}(d_{r,b},d_{\phi,b},d_{\dot r,b}),
\label{eq:obs_precision}
\end{equation}
where $d_{r,b}$, $d_{\phi,b}$, and $d_{\dot r,b}$ are the unit-illumination precision coefficients for range, azimuth, and radial velocity, respectively, and $\zeta_{b,t}=\iota_{b,t-1}$ is the precision scaling induced by the preceding target-directed illumination. The observed echo power $e_{b,t}$ and the binary detection flag $\chi_{b,t}$ follow the same timing. The detection flag follows an echo-dependent detection model, with a fixed false-alarm rate under blockage.

\section{Calibrated Availability Prediction}
\subsection{History-Aware Joint Pattern Predictor}
Because common blockers correlate links and the allocator constrains complete patterns, the predictor maps $\mathcal{H}_t$ to a joint probability vector over $S_{t+1}$. Let $L$ denote the history length. For BS $b$ and $\tau\in\{t-L+1,\ldots,t\}$, the feature is
\begin{equation}
\begin{aligned}
\mathbf{f}_{b,\tau}={}&
\big[\chi_{b,\tau},\widetilde e_{b,\tau},
\widehat{\mathrm{SCNR}}_{b,\tau},
r_{b,\tau},\phi_{b,\tau},\\
&\phantom{\big[}\dot r_{b,\tau},
\nu_{b,\tau},\rho_{b,\tau},
\iota_{b,\tau-1},\mu_{b,\tau}\big]^\mathsf{T}.
\end{aligned}
\end{equation}
Here, $(\chi_{b,\tau},\widetilde e_{b,\tau},
\widehat{\mathrm{SCNR}}_{b,\tau})$ summarize detection and echo
evidence, with $\widehat{\mathrm{SCNR}}_{b,\tau}$ denoting the
echo-derived SCNR estimate; $(r_{b,\tau},\phi_{b,\tau},
\dot r_{b,\tau})$ describe the estimated relative sensing geometry from
$\hat{\mathbf x}_{\tau|\tau}$ through $\mathbf g_b(\cdot)$;
$\nu_{b,\tau}$ and $\rho_{b,\tau}$ are the instantaneous and
smoothed normalized innovation squared (NIS), respectively; $\iota_{b,\tau-1}$ is the preceding
target-directed illumination; and $\mu_{b,\tau}$ is the
log-transformed consecutive-miss count. The normalized echo is
$\widetilde e_{b,\tau}=e_{b,\tau}/(\iota_{b,\tau-1}+\epsilon_0)$,
where $\epsilon_0>0$ prevents unstable division under weak
illumination.

The shared motion context is
\begin{equation}
\boldsymbol{\psi}_{t}=
\big[
\hat{\mathbf{x}}_{t|t}^\mathsf{T},
\lambda_1(\mathbf{P}_{t|t}),\ldots,\lambda_4(\mathbf{P}_{t|t}),
\widehat v_t,\widehat \theta_t
\big]^\mathsf{T},
\end{equation}
where
$\widehat v_t=\sqrt{\widehat v_{x,t}^{\,2}+\widehat v_{y,t}^{\,2}}$
and
$\widehat\theta_t=\operatorname{atan2}(\widehat v_{y,t},\widehat v_{x,t})$
are computed from the estimated velocity components.
The complete predictor input is $\mathcal{H}_t\triangleq(\{\mathbf{f}_{b,t-L+1:t}\}_{b=1}^{B},\boldsymbol{\psi}_t)$. The covariance eigenvalues encode tracking confidence, while speed and heading help separate local link loss from a trajectory change. A shared gated recurrent unit (GRU) with trainable parameters $\boldsymbol{\Theta}$ encodes each length-$L$ history:
\begin{equation}
\mathbf{c}_{b,t}=
\operatorname{GRU}_{\boldsymbol{\Theta}}
\left(\mathbf{f}_{b,t-L+1:t}\right).
\end{equation}
Shared weights capture common dynamics, while larger $L$ controls the temporal context available to the predictor. With $\bar{\mathbf{c}}_{t}=\sum_b\mathbf{c}_{b,t}/B$, the joint head outputs
\begin{equation}
\widehat{\boldsymbol{\pi}}_t
=\operatorname{softmax}
\left(f_{\boldsymbol{\Theta}}\left(
\mathbf{c}_{1,t},\ldots,\mathbf{c}_{B,t},
\bar{\mathbf{c}}_{t},
\boldsymbol{\psi}_t\right)\right)
\in\mathbb{R}^{2^B}.
\end{equation}
Its $s$-th component, $\widehat{\pi}_{t,s}$, estimates $\Pr\{S_{t+1}=s\mid\mathcal{H}_t\}$. The joint head preserves cross-link correlation. Because APS uses this ranking, training minimizes negative log-likelihood (NLL)
\begin{equation}
\mathcal L_{\rm pred}(\boldsymbol{\Theta})
=
-\frac{1}{n_{\rm tr}}
\sum_{i=1}^{n_{\rm tr}}
\log \widehat \pi_{i,y_i},
\end{equation}
where $n_{\rm tr}$ is the number of training samples and $y_i$ is the realized next-slot pattern label of sample $i$. The output is a ranked joint distribution rather than a hard link decision.

\subsection{Adaptive Prediction Set Calibration}
APS converts the ranking into a variable-size uncertainty set. Calibration samples are obtained from held-out episodes under a fixed All-serviceable reference policy. For calibration sample $i$, let $r_i$ denote the rank of the true label after sorting the predicted probabilities in descending order; its APS score is the cumulative probability up to and including this rank. The APS threshold $\widehat{c}_\alpha$ is the $\lceil(n_{\rm cal}+1)(1-\alpha)\rceil$-th ordered calibration score. At test time, patterns are sorted in decreasing probability and included in $\Gamma_{\alpha,t}$ up to and including the first one whose cumulative probability reaches $\widehat c_\alpha$. Under exchangeability between calibration and test scores generated by this reference policy,
\begin{equation}
\Pr\left\{S_{t+1}\in\Gamma_{\alpha,t}\right\}\geq 1-\alpha .
\label{eq:coverage}
\end{equation}
Deployed actions can shift histories, so coverage is policy-dependent. Since the set may include the all-blocked state, which power cannot improve, the serviceable prediction set and the outage warning flag are
\begin{equation}
\Gamma^{\rm svc}_{\alpha,t}
=\Gamma_{\alpha,t}\cap\mathcal{S}_{\rm svc}, \qquad o_{\alpha,t}=\mathbb{I}\{s_{\rm out}\in\Gamma_{\alpha,t}\}.
\label{eq:gamma_svc}
\end{equation}
$|\Gamma^{\rm svc}_{\alpha,t}|$ sets the robust constraint count, while $o_{\alpha,t}$ records a plausible power-insensitive outage. An empty serviceable set yields communication-only allocation.

\section{Serviceability-Conditioned Power Allocation}
\subsection{Fixed-Direction Transmission}
Fixing directions within each slot avoids joint beam--power optimization and keeps the FIM affine in scalar powers.
\begin{equation}
\mathbf{u}_{b,t}=
\frac{\mathbf{E}_b^\mathsf{H}\mathbf{a}_b(\hat{\mathbf{q}}_{t+1|t})}
{\left\|\mathbf{E}_b^\mathsf{H}\mathbf{a}_b(\hat{\mathbf{q}}_{t+1|t})\right\|_2},
\label{eq:sensing_direction}
\end{equation}
where $\hat{\mathbf{q}}_{t+1|t}=[\hat q_{x,t+1|t},\hat q_{y,t+1|t},h_{\rm u}]^\mathsf{T}$ is the one-step predicted UAV position, $\mathbf{E}_b$ selects BS $b$'s antennas and $\mathbf{a}_b(\cdot)$ is its steering vector. The sensing direction is refreshed after each EKF update. The unit-norm communication directions $\mathbf{v}_{k,t}$ are obtained by zero-forcing (ZF) or regularized zero-forcing (RZF) using the current communication CSI. Powers $p^{\rm c}_{k,t}\geq0$ and $p^{\rm s}_{b,t}\geq0$ control user service and sensing illumination:
\begin{equation}
\mathbf{R}_{t}=
\sum_{k=1}^{K}p^{\rm c}_{k,t}\mathbf{v}_{k,t}\mathbf{v}_{k,t}^\mathsf{H}
+\sum_{b=1}^{B}p^{\rm s}_{b,t}\mathbf{u}_{b,t}\mathbf{u}_{b,t}^\mathsf{H}.
\label{eq:covariance}
\end{equation}
Directions are recomputed each slot, while $p^{\rm s}_{b,t}=0$ realizes sensing handover through power redistribution without losing convexity.

\subsection{Communication Constraints}
For brevity, the common slot index $t$ is suppressed below. Define
$\mathbf{p}^{\rm c}=[p_1^{\rm c},\ldots,p_K^{\rm c}]^\mathsf{T}$,
$\mathbf{p}^{\rm s}=[p_1^{\rm s},\ldots,p_B^{\rm s}]^\mathsf{T}$,
and
$\mathbf{p}=[(\mathbf{p}^{\rm c})^\mathsf{T},
(\mathbf{p}^{\rm s})^\mathsf{T}]^\mathsf{T}$.
The sensing action must satisfy the communication quality-of-service (QoS) targets and per-BS power budgets. For user $k$, define the fixed-direction gains
\begin{equation}
g_{k\ell}=|\mathbf{h}_k^\mathsf{H}\mathbf{v}_{\ell}|^2,\qquad
g^{\rm s}_{kb}=|\mathbf{h}_k^\mathsf{H}\mathbf{u}_{b}|^2 .
\end{equation}
Here, $g_{kk}$ is the desired gain, $g_{k\ell}$ is multiuser-interference gain, and $g^{\rm s}_{kb}$ is leakage from sensing beam $b$. The downlink signal-to-interference-plus-noise ratio (SINR) requirement is
\begin{equation}
p_k^{\rm c}g_{kk}\geq
\gamma_k\left(
\sum_{\ell\neq k}p_{\ell}^{\rm c}g_{k\ell}
+\sum_{b=1}^{B}p_b^{\rm s}g_{kb}^{\rm s}
+\sigma_k^2
\right),\quad \forall k.
\label{eq:sinr}
\end{equation}
Here, $\gamma_k$ is the SINR target of user $k$ and $\sigma_k^2$ is its receiver noise power. Larger $\gamma_k$ improves communication protection but leaves less power for sensing. With fixed directions, \eqref{eq:sinr} is linear, as are the local budgets
\begin{equation}
\sum_{k=1}^{K}\eta^{\rm c}_{m,k}p_k^{\rm c}
+\sum_{b=1}^{B}\eta^{\rm s}_{m,b}p_b^{\rm s}
\leq P_m^{\max},\quad \forall m,
\label{eq:pbs}
\end{equation}
where $P_m^{\max}$ is the total transmit-power budget of BS $m$, $\eta^{\rm c}_{m,k}=\|\mathbf{E}_m\mathbf{v}_{k,t}\|_2^2$ and $\eta^{\rm s}_{m,b}=\|\mathbf{E}_m\mathbf{u}_{b,t}\|_2^2$ are BS-$m$ beam fractions.

\subsection{Pattern-Dependent Tracking Information}
Each availability pattern activates a different sensing geometry, so the same power vector must be evaluated realization-wise.
\begin{equation}
\iota_b(\mathbf{p})=
\sum_{k=1}^{K}p_k^{\rm c}
\left|\mathbf{a}_b^\mathsf{H}\mathbf{E}_b\mathbf{v}_{k}\right|^2
+\sum_{j=1}^{B}p_j^{\rm s}
\left|\mathbf{a}_b^\mathsf{H}\mathbf{E}_b\mathbf{u}_{j}\right|^2.
\label{eq:illumination}
\end{equation}
Both beam types illuminate the UAV according to their steering alignment. For pattern $s$, $\delta_b^{(s)}$ gates link $b$'s contribution to the sensing FIM:
\begin{equation}
\mathbf{J}^{(s)}(\mathbf{p})=
\sum_{b=1}^{B}\delta_b^{(s)}\iota_b(\mathbf{p})
\mathbf{H}_b^\mathsf{T}\mathbf{D}_b\mathbf{H}_b,
\label{eq:fim}
\end{equation}
Here, $\delta_b^{(s)}$ removes blocked links, while $\mathbf{H}_b^\mathsf{T}\mathbf{D}_b\mathbf{H}_b$ captures the geometry and unit-illumination precision of link $b$.

With $\mathbf{P}^{-}_{t+1}=\mathbf{P}_{t+1|t}$, combining prior and sensing precision gives the candidate posterior covariance
\begin{equation}
\mathbf{P}_{+}^{(s)}(\mathbf{p})=
\left[
(\mathbf{P}^{-}_{t+1})^{-1}+
\mathbf{J}^{(s)}(\mathbf{p})
\right]^{-1}.
\label{eq:posterior}
\end{equation}
Let $\epsilon_{\rm tr}$ be the nominal evaluation threshold and $\epsilon_{\rm plan}\leq\epsilon_{\rm tr}$ the allocator threshold. With $\mathbf{B}_q=[\mathbf{I}_2,\mathbf{0}]$ selecting position states, the required position covariance trace is
\begin{equation}
\operatorname{tr}\left(
\mathbf{B}_q\mathbf{P}_{+}^{(s)}(\mathbf{p})
\mathbf{B}_q^\mathsf{T}
\right)
\leq \epsilon_{\rm plan},
\qquad
\forall s\in\Gamma^{\rm svc}_{\alpha,t}.
\label{eq:tracking_constraint}
\end{equation}
Smaller $\epsilon_{\rm plan}$ increases robustness, power, and infeasibility risk. Since $\mathbf{J}^{(s_{\rm out})}=\mathbf{0}$, $s_{\rm out}$ is treated as outage.

\subsection{Serviceability-Conditioned Robust Power Allocation}
The allocator minimizes total power while satisfying communication and tracking for every retained serviceable pattern:
\begin{subequations}
\label{eq:main_problem}
\begin{align}
\min_{\mathbf{p}^{\rm c},\mathbf{p}^{\rm s}}\quad&
\sum_{k=1}^{K}p_k^{\rm c}+\sum_{b=1}^{B}p_b^{\rm s}\\
\mathrm{s.t.}\quad&
\eqref{eq:sinr},\ \eqref{eq:pbs},\ \eqref{eq:tracking_constraint}\\
&p_k^{\rm c}\geq0,\quad p_b^{\rm s}\geq0,\quad \forall k,b .
\end{align}
\end{subequations}
Robustness applies to $\Gamma_{\alpha,t}^{\rm svc}$, while an empty set yields communication-only allocation. For convexity, introduce a $4\times4$ auxiliary matrix $\mathbf{U}_s\succeq\mathbf{0}$ for each $s\in\Gamma_{\alpha,t}^{\rm svc}$ and apply the Schur complement \cite{boyd2004convex}:
\begin{subequations}
\begin{align}
&\begin{bmatrix}
(\mathbf{P}_{t+1}^{-})^{-1}+\mathbf{J}^{(s)}(\mathbf{p})
& \mathbf{I}_4\\
\mathbf{I}_4 & \mathbf{U}_s
\end{bmatrix}
\succeq \mathbf{0},
\qquad
\forall s\in\Gamma_{\alpha,t}^{\rm svc},
\\
&\operatorname{tr}\left(
\mathbf{B}_q\mathbf{U}_s\mathbf{B}_q^\mathsf{T}
\right)
\leq \epsilon_{\rm plan},
\qquad
\forall s\in\Gamma_{\alpha,t}^{\rm svc}.
\end{align}
\label{eq:schur_trace}
\end{subequations}
Since $(\mathbf{P}^{-}_{t+1})^{-1}+\mathbf{J}^{(s)}(\mathbf{p})\succ\mathbf{0}$, these linear matrix inequalities (LMIs) are equivalent to \eqref{eq:tracking_constraint}. Problem~\eqref{eq:main_problem} admits an SDP formulation with one Schur-complement LMI per retained serviceable pattern.

\begin{proposition}[Coverage-to-tracking reliability]
\label{prop:coverage_tracking}
For a fixed deployed policy, let $\mathcal{V}_{t+1}$ denote the event
that the realized posterior position trace exceeds the nominal
threshold $\epsilon_{\rm tr}$, and let $\mathcal{F}_t$ denote
the event that $\Gamma^{\rm svc}_{\alpha,t}\neq\emptyset$ but
\eqref{eq:main_problem} is infeasible. Define
$p_{\rm miss}=\Pr\{S_{t+1}\notin\Gamma_{\alpha,t}\}$,
$p_{\rm inf}=\Pr\{\mathcal{F}_t\cap
(S_{t+1}\in\mathcal{S}_{\rm svc})\}$, and
$p_{\rm out}=\Pr\{S_{t+1}=s_{\rm out}\}<1$.
Under the covariance model \eqref{eq:posterior}, applying the
robust solution whenever feasible yields
\begin{align}
R_{\rm svc}
&\triangleq \Pr\{\mathcal{V}_{t+1}^{c}\mid
S_{t+1}\in\mathcal{S}_{\rm svc}\}
\geq \left[1-
\frac{p_{\rm miss}+p_{\rm inf}}
{1-p_{\rm out}}\right]_{+},
\label{eq:svc_reliability_bound}\\
A_{\rm all}
&\triangleq p_{\rm out}
+(1-p_{\rm out})(1-R_{\rm svc})\notag\\
&\leq \min\!\left\{1,\,
p_{\rm out}+p_{\rm miss}+p_{\rm inf}\right\}.
\label{eq:all_acct_bound}
\end{align}
Here, $[x]_{+}=\max\{x,0\}$. Calibration--deployment exchangeability gives $p_{\rm miss}\leq\alpha$. Otherwise, $p_{\rm miss}$ is the deployed-policy miscoverage. Feasibility on every serviceable slot gives $p_{\rm inf}=0$.
\end{proposition}

\begin{proof}
For a retained serviceable realization outside $\mathcal{F}_t$, the constraint ensures a trace no larger than $\epsilon_{\rm plan}\leq\epsilon_{\rm tr}$ and thus rules out $\mathcal{V}_{t+1}$. Thus, serviceable violation implies miscoverage or infeasibility. The union bound and conditioning prove \eqref{eq:svc_reliability_bound}. Adding outage proves
\eqref{eq:all_acct_bound}.
\end{proof}

The per-slot control loop is summarized in Algorithm~\ref{alg:control}.

\begin{algorithm}[t]
\footnotesize
\DontPrintSemicolon
\KwIn{$\mathcal{H}_t$, $\mathbf{R}_{t-1}$, communication CSI, and EKF posterior}
\KwOut{current transmit action $\mathbf{R}_{t}$}
Update the EKF from observations generated by $\mathbf{R}_{t-1}$; form features\;
Predict $\widehat{\boldsymbol{\pi}}_t$; construct $\Gamma_{\alpha,t}$ by APS and split it using \eqref{eq:gamma_svc}\;
\lIf{$o_{\alpha,t}=1$}{issue a predicted sensing-outage warning}
Form ZF/RZF communication and predicted-position sensing directions\;
\lIf{$\Gamma^{\rm svc}_{\alpha,t}=\emptyset$}{select a communication-only action}
\lElse{solve \eqref{eq:main_problem}; use the robust solution if feasible, otherwise use the worst-trace fallback}
Apply $\mathbf{R}_{t}$; it affects the radar observations at slot $t+1$\;
\caption{Outage-aware per-slot control}
\label{alg:control}
\end{algorithm}

\section{Simulation Results}
\label{sec:experiments}

\subsection{Setup and Metrics}
\label{subsec:setup}
Results are averaged over five seeds in a geometry-based $1000\,\mathrm{m}\times1000\,\mathrm{m}$ scene with $B=3$ ISAC BSs at $(80,160,30)$, $(920,180,30)$, and $(500,900,30)$~m, $K=2$ communication users, and a UAV flying at an altitude of $120$ m. Rectangular blockers, distance attenuation, and the SCNR threshold in \eqref{eq:availability} determine physical availability. As a simulation simplification, latent availability gates EKF updates; echo noise and detection errors affect policy histories. Held-current and Markov are given the same latent information for a matched comparison. For each seed, the dataset contains 300, 60, 60, and 50 episodes for training, validation, calibration, and testing, respectively, with 60 slots per episode.

All policies use matched random draws within each test episode. Unless otherwise stated, $\alpha=0.1$, $L=12$, the prediction horizon is one slot, the SINR target is $5$ dB, the per-BS power budget is $P_m^{\max}=100$ in dimensionless simulation units, and $(\epsilon_{\rm plan},\epsilon_{\rm tr}) =(0.45,0.5)~\mathrm{m}^2$. Violations are evaluated using $\epsilon_{\rm tr}$. Performance is measured by the serviceable-slot violation $1-R_{\rm svc}$ and the outage-inclusive accounting rate $A_{\rm all}$.

\subsection{Prediction and Calibration}

Prediction is sufficiently concentrated for APS to form compact uncertainty sets, while calibration is empirically conservative. On held-out validation episodes replayed under the fixed All-serviceable reference policy, the GRU achieves 98.4\% top-2 accuracy, and APS retains 3.0 patterns on average while attaining 99.96\% empirical coverage at the nominal 90\% level. Coverage also remains above 99.3\% under deployed policies. Thus, in the tested regime, APS mainly controls how much plausible ambiguity is passed to the allocator, trading robustness overhead against efficiency rather than compensating for frequent prediction misses.

\subsection{Main Policy Comparison}

\begin{table}[!t]
\caption{Closed-loop cost and feasibility at $B=3$, $K=2$, $\alpha=0.1$.}
\label{tab:main_results}
\centering
\scriptsize
\setlength{\tabcolsep}{3.1pt}
\begin{tabular}{lcccc}
\hline
Policy & $A_{\rm all}$ & Fallback$^\ast$ & Avg. power & Avg. patterns\\
\hline
Held-current & $9.46{\pm}0.73\%$ & $0.88\%$ & $12.20{\pm}0.46$ & $0.94{\pm}0.01$\\
Markov & $7.70{\pm}0.72\%$ & $1.28\%$ & $14.66{\pm}0.23$ & $1.65{\pm}0.01$\\
MAP & $18.01{\pm}2.27\%$ & $0.81\%$ & $12.03{\pm}0.57$ & $0.92{\pm}0.03$\\
Expected-FIM & $13.39{\pm}0.96\%$ & $0.92\%$ & $13.40{\pm}0.34$ & $0.00{\pm}0.00$\\
Conformal-hard & $83.62{\pm}4.01\%$ & $83.68\%$ & $6.50{\pm}0.83$ & $5.24{\pm}0.35$\\
Conformal-serviceable & $5.74{\pm}0.71\%$ & $1.69\%$ & $20.23{\pm}0.37$ & $3.44{\pm}0.05$\\
All-serviceable & $5.73{\pm}0.71\%$ & $4.82\%$ & $38.14{\pm}0.35$ & $7.00{\pm}0.00$\\
Oracle & $5.74{\pm}0.71\%$ & $0.89\%$ & $12.05{\pm}0.45$ & $0.94{\pm}0.01$\\
\hline
\end{tabular}
\parbox{0.98\columnwidth}{\scriptsize \emph{Notes:} The policies use the current pattern, Markov top-2, MAP, probability-weighted FIM, full APS, serviceable APS, all-serviceable, and true patterns, respectively. Avg. power is the slot-averaged total transmit power in the dimensionless simulation units defined in Sec.~\ref{subsec:setup}. Avg. patterns is the average number of pattern-wise constraints. Expected-FIM has no pattern-wise constraints and is therefore reported as $0$. $^\ast$On infeasibility, fallback minimizes worst trace over protected patterns plus $10^{-3}$ total power, retaining SINR and BS budgets. Entries with $\pm$ report the mean $\pm$ standard error of the mean (SEM) over five seeds.}
\end{table}

\begin{figure}[!t]
\centering
\includegraphics[width=\linewidth]{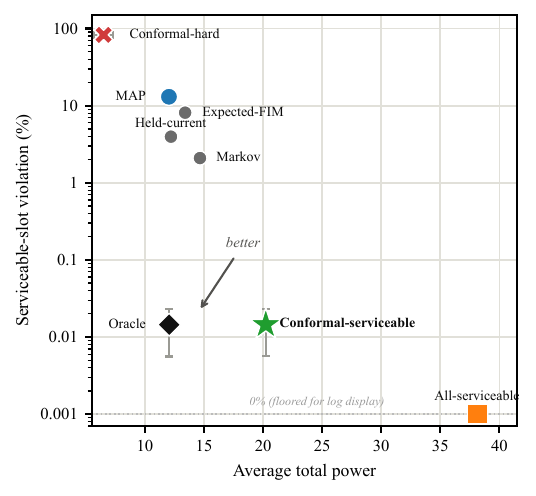}
\caption{Power and reliability for the policies in Table~\ref{tab:main_results}. Error bars show $\pm$ SEM, and the lower-left region is preferred. The exact-zero All-serviceable point is floored at $0.001\%$ for the log scale.}
\label{fig:policy_frontier}
\end{figure}

Table~\ref{tab:main_results} and Fig.~\ref{fig:policy_frontier} reveal a clear robustness–efficiency ordering. Point or averaged surrogates under-protect uncertain future states: MAP and Markov act on only a small number of predicted patterns, while Expected-FIM can mask low-information realizations through averaging. At the opposite extreme, Conformal-hard over-protects by treating the zero-FIM all-blocked state as enforceable, which drives frequent fallback. Serviceability conditioning removes this non-actionable state while retaining set-valued protection over plausible serviceable patterns, placing the proposed method close to the outage floor. The horizontal gaps in Fig.~\ref{fig:policy_frontier} have a direct interpretation: the gap to Oracle is the price of acting before the true pattern is revealed, whereas the gap to All-serviceable is the benefit of pruning implausible patterns. The latter reduces power by 47.0\%.

The baseline failures support this interpretation. MAP suppresses useful sensing when a top-1 error crosses the serviceability boundary, which occurs on $5.0\%$ of the serviceable slots. Conformal-hard fails for the opposite reason: once the uncontrollable all-blocked state enters its robust set, no power can restore the trace target, producing an $83.68\%$ fallback rate. The proposed serviceability split avoids both mismatches by protecting uncertainty only where power can change the sensing information.

\subsection{Risk, Shift, and Scalability}

Increasing $\alpha$ reduces the protected set and required power, while serviceable-slot violation remains nearly unchanged and below $0.03\%$. Thus, in the current overcovered regime, $\alpha$ mainly controls efficiency rather than reliability. With the nominal predictor and APS threshold frozen, six shifts (BS layout, denser blockers, faster blockers, UAV turns, higher detection thresholds, and noisier SCNR) yield below $0.1\%$ violation for Conformal-serviceable versus up to $23.86\%$ for MAP. 
\begin{figure}[!t]
\centering
\includegraphics[width=\linewidth]{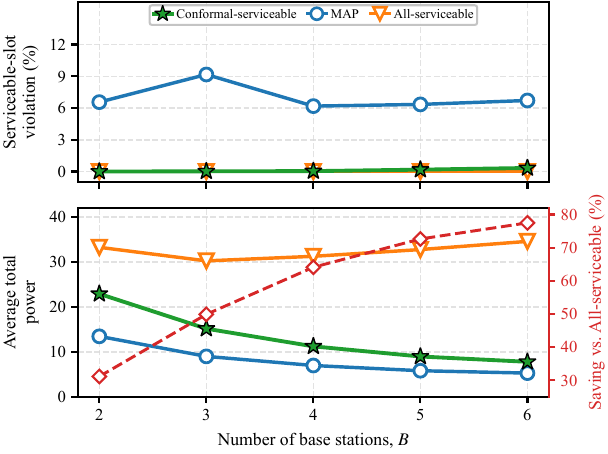}
\caption{Scaling with $B$ under a circular BS layout: serviceable-slot violation and power with savings over All-serviceable.}
\label{fig:b_scalability}
\end{figure}

For the scalability experiment in Fig.~\ref{fig:b_scalability}, we use a separate circular BS layout with equally spaced stations, while retaining the other parameters in Sec.~\ref{subsec:setup}. Fig.~\ref{fig:b_scalability} shows that increasing $B$ lowers Conformal-serviceable power and enlarges its saving over All-serviceable, with only a modest change in violation. At $B=6$, only $7.53$ constraints are retained on average for $63$ serviceable patterns, indicating that selective robustification avoids the full combinatorial set. For larger networks, the $2^B$ predictor outputs and one LMI per retained pattern can increase computation and memory costs.

\section{Conclusion}
This paper proposed a serviceability-conditioned conformal robust allocation framework for UAV tracking under dynamic blockage. The framework calibrates joint availability prediction sets and separates serviceable patterns from physical sensing outage. For fixed per-slot beam directions, the FIM-constrained scalar-power allocation is formulated as an SDP, while the reliability analysis attributes serviceable-slot failures to prediction-set miscoverage or optimization infeasibility and accounts separately for physical outage. Future work will investigate calibration under policy shift and structured prediction for larger networks.

\bibliographystyle{IEEEtran}
\bibliography{refs}

\end{document}